\documentclass[fleqn,12pt]{article}
\usepackage{tikz}
\usepackage{amsmath,amssymb,amsthm,ctable,enumitem,fullpage,pgfplots,setspace,subdepth,verbatim}
\pgfplotsset{compat=1.17}
\usepackage[round]{natbib}
\usepackage{accents}
\newcommand{\ubar}[1]{\underaccent{\bar}{#1}}

\newtheorem{theorem}{Theorem}[section] 
\numberwithin{theorem}{section}

\newtheorem{observation}{Observation}

\theoremstyle{definition}
\newtheorem{definition}{Definition}
\theoremstyle{remark}
\newtheorem{remark}{Remark}
\newtheorem{example}{Example}

\begin{document}

\title{Smart contracts and reaction-function games%
\thanks{%
    We thank Erik Ansink, Hans Carlsson, Harold Houba, Peter Bro Miltersen, Alexandros Rigos, Nikolaj I. Schwartzbach, Mich Tvede, Peyton Young,  Lars Peter {\O}sterdal and seminar participants at Copenhagen Business School, Adam Smith Business School, Guangzhou University, and Chinese University of Hong Kong for valuable comments.
    Financial support from the Independent Research Fund Denmark (grant no. 4260-00050B) is gratefully acknowledged.}
}
\author{Jens Gudmundsson$^{\text{c}}$ \and Jens Leth Hougaard}
\date{{\small %
    Department of Food and Resource Economics, University of Copenhagen, Denmark \\
    $^\text{c}$Corresponding author: Rolighedsvej 23, 1958 Frederiksberg C, Denmark; jg@ifro.ku.dk}
    \\[2ex] \today} 
\maketitle

\begin{abstract}
    Blockchain-based smart contracts offer a new take on credible commitment, where players can commit to actions in reaction to actions of others.
    Such \emph{reaction-function games} extend on strategic games with players choosing reaction functions instead of strategies.
    We formalize a solution concept in terms of fixed points for such games, akin to Nash equilibrium, and prove equilibrium existence.
    Reaction functions can mimic ``trigger'' strategies from folk theorems on infinitely repeated games---but now in a one-shot setting---for instance to support Pareto-improvements on Nash equilibrium outcomes. %
    In some games, this can even be done through risk-free, \emph{safe} reaction functions.
    We apply our theoretical framework to symmetric investment games, which includes two prominent classes of games, namely \emph{weakest-link} and \emph{public-good} games.
    In both cases, we highlight a particular safe and optimal reaction function.
    In this way, our findings highlight how blockchain-based commitment can help overcome trust and free-riding barriers.
\end{abstract}

\newpage
\section{Introduction}

Smart contracts are digital tools that allow agents to coordinate in a decentralized environment \citep[]{Gans2021, BakosHalaburda2023, GudmHougRecip, GUdHouKo2024,GudmundssonHougaardSethuraman25}.
In short, a smart contract is a piece of computer code with a unique address on a blockchain.
Agents with shared access to the blockchain can deploy such contracts, deposit funds to them, and hard-code how these funds are used in reaction to choices made by others.
As such, an agent can let a contract act automatically on her behalf, in effect making contract design a strategic choice.
Once deployed, a contract cannot be changed.
In this way, smart contracts ensure commitment as strong as the cryptographic security of the network itself \citep[e.g.][]{LiangEtAl2024,Raskin2017}.

Smart contracts can thus serve as commitment devices in one-shot interactions:
rather than choosing an action directly, a player can commit to a conditional action plan that prescribes what to do as a function of what others do.
We model such plans as \emph{reaction functions} (i.e., mappings from the actions chosen by the others to your own actions) and study the resulting \emph{reaction-function games}.
As such, we follow up on a line of literature that mixes ideas from game theory and computer science, dating back to von Neumann and further developed through ``metagames'' \citep{Howard1971}, ``program equilibrium'' \citep{Tennenholtz2004}, and ``commitment games'' \citep{Kalaietal2010}.
Once players commit to reaction functions, an outcome is implementable when the commitments are mutually consistent:
that is, the realized action profile is a fixed point of the reaction-function profile.
A profile of reaction functions is a \emph{reaction-function equilibrium} if no player can benefit by unilaterally changing her committed reaction function, taking into account the resulting fixed-point outcomes (the formal definition is found in Section~\ref{SEC:gamesDefs}).

Our main message is positive:
we show that reaction-function play allows players to reach better outcomes.
As smart-contract commitment is particularly viable in investment problems, made concrete through escrowed deposits and automated transfers, we devote particular attention to \emph{symmetric investment games}.
If investments are strategic complements (players want to invest more the more others invest), then reaction functions aid in solving coordination and trust issues;
we illustrate this in \emph{weakest-link games} \citep[e.g.][]{Riedletal2016}.
On the other hand, if free-riding incentives are present (players want others to invest but prefer not to invest themselves), then reaction functions open for a form of conditional collaboration, this time illustrated in \emph{public-good games} \citep[e.g.][]{BergstromEtAl1986}.

In addition, we provide the following theoretical foundations: 
\begin{enumerate}[noitemsep]
    \item We establish existence of equilibrium in every reaction-function game.
    This is possible without having to invoke, for instance, mixed strategies.
    \item In two-player games, we characterize the outcomes that can be supported in reaction-function equilibrium.
    This provides a one-shot analogue of the familiar ``folk theorem'' logic.
    The key insight is that smart-contract commitment allows implementation of Pareto-efficient outcomes that are not Nash equilibria in the conventional strategic game.
    \item With more than two players, the set of supportable outcomes becomes very large:
    reaction functions offer, in some sense, too much flexibility.
    In response, we develop a refinement based on \emph{safe play}.
    In spirit, this rules out ``implausible'' commitments, much like equilibrium refinements rule out non-credible threats.
\end{enumerate}

In conventional strategic games, a player can always ``play safe'' by choosing a maxmin action that guarantees the highest possible worst-case payoff.
The downside is that this may yield overly conservative outcomes.
For instance, in the \emph{Prisoner's dilemma}, safe play leads to the unique Pareto-dominated Nash equilibrium.
In reaction-function games, we can analogously restrict attention to safe reaction functions in the sense that, no matter the action profile of the others, the player commits to react in a way that ensures herself at least maxmin payoff.
But reaction functions provide more flexibility:
one can simultaneously play safe \emph{and} support high-payoff outcomes.
We find that, if safe play is a Nash equilibrium in the strategic game, then we can also support all potential Pareto improvements of this outcome in safe reaction-function equilibrium. 
For instance, Pareto-efficient cooperation can be made a safe reaction-function equilibrium in the \emph{Prisoner's dilemma}.

To emphasize the potential of smart-contract commitment, we hone in on symmetric investment games.
These provide a natural ground for application of our theory.
Smart contracts make conditional investment commitments concrete through escrowed deposits and automated transfers;
``safe play'' becomes simple and interpretable;
and investment choices have a natural order (higher investment is ``more'').
Accordingly, focusing on monotone reaction functions both captures economically plausible behavior and makes fixed-point existence and computation straightforward.
We leverage this structure and introduce two novel conditions on reaction functions---\emph{norm-proofness} and \emph{payoff consistency}---to obtain sharp, practical recommendations for two prominent classes of investment games.

In \emph{weakest-link games}, each player wants to match the minimum investment of the others, and the higher this is, the better for all \cite[e.g.][]{Riedletal2016}.
There appears to be an obvious solution---all invest the maximum amount---yet there is overwhelming experimental evidence that players fail to coordinate on the efficient equilibrium \citep{vanHuycketal1990,Battalioetal2001,Weberetal2001,BrandtsCooper2006,Weber2006}, with some rare exceptions \citep{EngelmannNormann2010,Riedletal2016,Riedletal2021}.
The mere belief or expectation that someone else will make the safe minimum investment is all it takes to justify making the minimum investment yourself \citep[e.g.][]{Aumann1990}.
Here, reaction-function commitment offers a straightforward resolution:
each player commits to match the minimum investment of the others, which includes a promise of high investment if all others do the same.
Interestingly, this type of ``uniform commitment'' has already found support elsewhere in the literature, for instance on climate negotiations \citep{Weitzman2014,SchmidtOckenfels2021}.
With ``high-risk'' parameters, we can go yet further:
the same reaction function is then singled out as weakly dominant against monotone play.
Hence, reaction-function play provides a very compelling way to overcome the trust issues inherent in the weakest-link game.

In \emph{public-good games}, each player's investment benefits everybody but is a net loss to the player herself \citep[e.g.][]{BergstromEtAl1986,Falkingeretal2000,Croson2007}.
Again, coordination around investing the maximal amount is welfare maximizing, but choosing to invest the minimum amount is safe and now even dominant.
Here, we identify reaction functions in which each player matches the average (rounded up or down) of the others' investments.
We argue further that rounding down is particularly appealing:
it is safe for all parameter values and supports the high-investment outcome in a symmetric reaction-function equilibrium.
Moreover, this type of monotone commitment makes free-riding self-defeating:
if all others match the average rounded down, then a unilateral deviation to never invest collapses the profile to the zero-investment fixed point.

\emph{Implementation sketch.}
We briefly touch on how to take reaction-function theory to practice using smart-contract investment commitments.
At a high level, each player deploys a contract detailing her reaction function and escrows a deposit that bounds her maximum investment.
In this way, we combine commitment (the reaction function cannot be changed after deployment) with automated execution (escrowed funds can be transferred without further interaction).
A simple coordinating contract can then
\begin{enumerate}[noitemsep]
    \item collect commitments and deposits,
    \item compute a fixed point of the submitted reaction functions, and
    \item execute the implied investments (returning unused deposits).
\end{enumerate}
The main practical bottleneck is the fixed-point computation.
However, for economic applications with natural classes of reaction functions (e.g., monotone functions on a finite lattice), simple iterative procedures will suffice and quickly pin down a fixed point.

\emph{Related literature.}
Our approach is related to a strand of literature in which players choose commitment devices or programs rather than actions, including \emph{metagames} \citep{Howard1971}, \emph{program equilibrium} \citep{Tennenholtz2004,MondererTennenholtz2009,Oesterheld2019,CooperEtAl2025}, and \emph{commitment games} \citep{Kalaietal2010}.
In these approaches, a player's ``commitment device'' typically specifies what to do as a function of the devices chosen by others \citep[see e.g.][Section~5, for details]{Kalaietal2010}.
As a consequence, even when the underlying action spaces are simple, the induced strategy space of devices is extremely rich.
The objects to be chosen are very large---even in a two-player, two-action game, there will be an infinite number of devices/programs, and each device is itself made out of an infinite number of choices (what to do against every possible device).
For contrast, our reaction functions condition plainly on \emph{actions} of others;
in a two-by-two game, there are only four reaction functions.
In this way, strategic choices are kept parsimonious and interpretable, and practical implementation through simple contracts becomes tractable.

Reaction functions are also superficially related to the experimental ``strategy method'' \citep{Selten1967,Seltenetal1997,BrandtsCharness2011}, there sometimes known as ``contribution tables'' \citep{Fischbacheretal2001,ThoniVolk2018}.
However, in this field, the goal is merely to \emph{elicit} reaction functions (e.g. to get a more comprehensive view of the experimental subject's though process and test ``off-path'' behavior)---there is no need to play them out against each other.
To illustrate the methodological difference, in \citet{Fischbacheretal2001} all subjects submit a reaction function \emph{and} an unconditional strategy (corresponding to a constant reaction function), after which one player is drawn at random and the realized outcome is uniquely determined (compare Remark~\ref{REM:uniqueness}) by that player's reaction to the others' unconditional strategies. 
The method generates more data to analyze, but issues such as fixed-point computation, existence, and multiplicity are irrelevant.
For contrast, these issues are at the core of our analysis.
An exception is \citet{Oechssleretal2022}, who develop a theoretical model and test its implications experimentally.
They explore a repeated public-good game and limit to monotone two-step reaction functions (circumventing any fixed-point issues);
see also \citet{ReischmannOechssler2018}.

Finally, there is a growing literature on smart contracts as commitment devices in dynamic games.
\citet{Brzustowski-AER} show that, giving a monopolist access to such credible dynamic contracts overturns the classic Coase conjecture and makes it possible to achieve payoff bounded away from the buyer's lowest valuation for any choice of discount factor.
Relatedly, so-called ``Stackelberg games'' have recently gained new attention in the computer science literature.
Here, the follower can strategically deploy a reaction function that credibly commits to irrational responses in their subgame, thereby influencing the leader's optimal choice of action to the follower's own advantage \citep[e.g.][]{Schwartzbach}.

\emph{Outline.}
In Section~\ref{SEC:gamesDefs}, we introduce reaction-function games and compare them to conventional strategic games.
In Section~\ref{SEC:supportedOutcomes}, we explore the set of outcomes that can be supported in reaction-function equilibrium.
In Section~\ref{SEC:safe}, we develop the notion of ``safe'' reaction functions.
In Section~\ref{SEC:investmentGames}, we present our main application to symmetric investment games.
Section~\ref{SEC:conclusions} closes with concluding remarks.
Proofs are postponed to the Appendix.

\section{Games and definitions} \label{SEC:gamesDefs}

In this section, we introduce \emph{reaction-function games} and relate them to conventional \emph{strategic games}.
In strategic games, formalized in Subsection~\ref{SUB:SG}, players individually select actions that jointly result in an outcome over which players have preferences. 
The classic solution concept is Nash equilibrium. 
In reaction-function games, players individually specify their reaction to the actions of other players and are able to commit to these reactions. 
Determining the outcome of a reaction-function profile is now more intricate and consequently defining the game and a natural equilibrium concept requires further elaboration.
This is covered in Subsections~\ref{SUB:RF}--\ref{SUB:RFE}.
We assume players have complete information on the game.

\subsection{Strategic games} \label{SUB:SG}

There is a finite set of $n$ \textbf{players}~$N$.
Each player $i$ has a finite set of at least two \textbf{actions}~$A_i$. 
An \textbf{outcome} $a \in A \equiv \times_i A_i$ is an action profile.
Each player $i$'s preference is represented by a \textbf{payoff function} $u_i \colon A \to \mathbb{R}$, where $i$ derives payoff $u_i(a)$ at outcome~$a$. 
The triple $(N,A,(u_i)_i)$ defines a strategic game.
A \textbf{Nash equilibrium} is an action profile $a$ such that, for each player~$i$ and action~$a'_i \in A_i$, $u_i(a) \geq u_i(a'_i,a_{-i}) $.%
\footnote{As usual, the subscript $-i$ projects onto the subspace relative to $N \setminus \{i\}$ and $a = (a_i, a_{-i})$.}
Let $A_{-i} \equiv \times_{j \neq i} A_j$.
Note that existence of a Nash equilibrium is not guaranteed as we do not consider mixed strategies.

\subsection{Reaction functions} \label{SUB:RF}

In reaction-function games, players choose \emph{reaction functions} rather than actions.
A \textbf{reaction function} $R_i \colon A_{-i} \to A_i$ specifies the action $R_i(a_{-i}) \in A_i$ that player~$i$ commits to play in reaction to the actions $a_{-i} \in A_{-i}$ of the other players.
That is, players now map out a complete plan of what to do contingent on the actions of the others.
Let $\mathcal{R}_i$ be the set of player $i$'s reaction functions.
As action sets are finite, so is the set of reaction functions~$\mathcal{R}_i$. 
A profile of reaction functions, a \textbf{profile} for short, is $R \in \mathcal{R} \equiv \times_{i \in N} \mathcal{R}_i$.
For $a \in A$, let $R(a) \in A$ be the action profile in which each player $i$ plays action $R_i(a_{-i})$. Letting players choose reaction functions enlarges the strategy space considerably compared to strategic games: there are $\left\vert{\mathcal{R}}\right\vert = (m^n)^n$ reaction-function profiles versus 
$\left\vert{A}\right\vert = m^n$ outcomes (i.e., action profiles). 

If players are restricted to degenerate, \emph{constant} reaction functions, then we recover strategic games as defined in Subsection~\ref{SUB:SG}.
These comprise choosing the same action $a_i$ irrespective of the choices $a_{-i}$ of the others:
set $R_i(a_{-i}) = a_i$ for each $a_{-i} \in A_{-i}$.
In economic applications with decision variables such as quantities and prices, it may be natural to assume an order $\leq_i$ on each action set~$A_i$.
Constant reactions then belong to the larger class of \emph{monotone} (non-decreasing) reactions for which $a_{-i} \leqq a'_{-i}$ implies $R_i(a_{-i}) \leq R_i(a'_{-i})$, where $a_{-i} \leqq a'_{-i}$ whenever $a_j \leq_j a'_j$ for each player $j \neq i$.
For example, in games with strategic complements \citep[e.g.][]{Topkis1979,Vives1990, Echenique2004}, where an increase in one player's action makes the other players want to increase theirs, the \textit{best-reply} reaction $\textit{BR}_i(a_{-i}) \in \arg \max_{a_i \in A_i} u_i(a_i,a_{-i})$ is monotone.

\subsection{Fixed points in reaction-function games} \label{SUB:preferences}

Reaction functions are commitments to certain actions in reaction to the actions of the others.
In this way, an outcome can be realized if and only if player commitments align:
for each player~$i$, the action $a_i$ must match $i$'s committed reaction to the actions $a_{-i}$ of the others.
Put differently, outcome $a$ is a fixed point of the profile $R$ in the sense that, for each player $i$, $a_i = R_i(a_{-i})$.
Let $\mathcal{E}(R) = \{ a \in A : R(a) = a \}$ be the set of fixed points of $R$ (which may be empty).
Our approach builds on the following two assumptions:
\begin{enumerate}[noitemsep]
    \item Player $i$ evaluates profile $R$ by $i$'s preferred fixed point at~$R$.
    \item All fixed-point outcomes are preferred to the case in which there is no fixed point.
\end{enumerate}
More specifically, we extend the definition of payoffs to reaction-function profiles as follows.
For each player $i$ and profile~$R$, let
\[
    U_i(R) = \begin{cases}
        \max_{a \in \mathcal{E}(R)} u_i(a) & \text{if $\mathcal{E}(R) \neq \emptyset$} \\
        - \infty & \text{if $\mathcal{E}(R) = \emptyset$.}
    \end{cases}
\]

The first assumption captures player optimism, familiar for instance from games with externalities \citep[e.g.][]{AumannPeleg1960,Shenoy1979}.
The second addresses profiles $R$ for which $\mathcal{E}(R) = \emptyset$.
Such ``no fixed point'' profiles can arise (e.g., best-reply reactions in the game of \emph{Matching pennies}, see Table~\ref{TAB:MP} later) and we need a convention for how players evaluate them.
Our choice $U_i(R) = - \infty$ should be understood purely as a normalization:
it means that a profile without fixed points cannot be implemented and is therefore ranked below profiles with fixed points.
Importantly, $- \infty$ does not represent an actual ``punishment'' imposed by any player or by the mechanism;
it is simply a convenient way to encode that, absent a fixed point, the model does not specify a realized outcome.
In some contexts, this ranking is perfectly suitable:
for instance, in a surplus-sharing or bargaining setting, any implementable agreement may be preferred by all to failing to agree.
The assumption is not needed to establish equilibrium existence (Theorem~\ref{TH:existence}) nor is it necessary if one instead imposes suitable restrictions on the permissible reaction functions (compare Section~\ref{SEC:investmentGames}).

For some classes of reaction functions, fixed points always exist.
For instance, if each reaction $R_i$ is monotone, %
then $R$ is a monotone function on a complete lattice and the set of fixed points $\mathcal{E}(R)$ is non-empty \citep{Tarski1955}.
Remark~\ref{REM:uniqueness} takes this one step further and explores fixed-point uniqueness. 

\begin{remark}[Sufficient condition for unique fixed point] \label{REM:uniqueness}
    We can associate to each profile $R$ a directed graph $(N,E)$ on nodes $N$ in which there is an edge $ij \in E$ whenever $i$'s reaction $R_i$ depends on (i.e., is not constant in) $j$'s action~$a_j$.
    If the graph is acyclic, then there is a unique fixed point.
    This applies for instance if there is some choice sequentiality or player hierarchy that makes a player's reaction depend only on the actions of those who follow (or precede) the player.
    Constant reactions produce empty (so acyclic) graphs;
    also the construction leading up to Theorem~\ref{TH:existence} will be acyclic.
    For cyclical graphs, there can be any number of fixed points (possibly none);
    see for instance the proof of Theorem~\ref{TH:moreplayers}.
    \hfill $\circ$
\end{remark}

To summarize, a reaction-function game $(N, \mathcal{R}, (U_i)_i)$ comprises a set of players~$N$.
Each player $i$ has access to a set of reaction functions $\mathcal{R}_i$ and is endowed with a payoff function $U_i \colon \mathcal{R} \to \mathbb{R} \cup \{ -\infty \}$.
This is induced by the payoffs $u_i$ on $A$ and specifically determined by $i$'s preferred fixed point when such a point exists.

\subsection{Reaction-function equilibrium} \label{SUB:RFE}

Next, we develop our solution concept for reaction-function games.
As players evaluate profiles by their subjectively preferred fixed point, different players may evaluate the same profile by different outcomes.
We will rule out this in equilibrium.
That is, an equilibrium profile $R$ will be \emph{unambiguous}, meaning that there is a Pareto-superior fixed point in $\mathcal{E}(R)$ and hence a clear common outcome to associate to the profile.
This is not a restriction on what players can play but rather a strengthening of the solution concept to be defined next (akin to restricting to symmetric equilibria in symmetric games).%
\footnote{There are other ways to achieve the same effect:
for instance, one could assume players have agreed beforehand on a ``selection function'' that mechanically selects a fixed point from $\mathcal{E}(R)$ or explicitly model player beliefs about what outcome will be realized.
Our approach is simple, interpretable, and quite restrictive;
the other approaches could lead to yet more equilibria.
As we shall see, equilibrium multiplicity is already a concern with our stronger assumption.}

\begin{definition}
    Profile $R \in \mathcal{R}$ is \textbf{unambiguous} if there is an outcome $a \in \mathcal{E}(R)$ such that, for each outcome $a' \in \mathcal{E}(R)$ and player~$i$, $u_i(a) \geq u_i(a')$.
\end{definition}

For example, the profile in which each player matches the action of the other player is unambiguous in the \emph{Prisoner's dilemma} (Table~\ref{TAB:PD}) but not in the \emph{Battle of the sexes} (Table~\ref{TAB:BoS}).

\begin{table}[!htb]
    \renewcommand{\arraystretch}{1.2}
    \centering
    $\begin{array}{r | cc}
          &  x  &  y  \\ \hline
        x & \underline{1},\underline{2} & 0,0 \\
        y & 0,0 & \underline{2},\underline{1}
    \end{array}$
    \caption{Payoff matrix for \emph{Battle of the sexes}.
    The ``match other'' reactions, which coincide with best replies, are highlighted.
    The two fixed points are not Pareto ordered, so the profile is not unambiguous and, in consequence, not a reaction-function equilibrium.}
    \label{TAB:BoS}
\end{table}

We take inspiration from Nash equilibrium as defined in Subsection~\ref{SUB:SG} to formulate a generalized solution concept for reaction-function games.
In such an equilibrium, no unilateral reaction-function change is beneficial. 
The restriction to unambiguous profiles aligns with the Nash equilibrium requirement that players’ beliefs about the consequences of their actions are correct:
all players associate to an equilibrium profile $R$ the same outcome $a \in \mathcal{E}(R)$.

\begin{definition}
    A \textbf{reaction-function equilibrium} is an unambiguous profile $R \in \mathcal{R}$ such that, for each player $i$ and reaction function $R'_i \in \mathcal{R}_i$, $U_i(R) \geq U_i(R'_i, R_{-i})$.
\end{definition}

If players are restricted to constant reaction functions, then fixed points are unique for each profile (Remark~\ref{REM:uniqueness}) and reaction-function equilibrium is logically equivalent to Nash equilibrium. 
We record this in Observation~\ref{OBS:constant}.
Applied to the \emph{Battle of the sexes}, there are two such constant reaction-function equilibria:
both players always play $x$, or both always play~$y$.%
\footnote{For\label{FN:BoS} the \emph{Battle of the sexes}, we find the striking conclusion that any outcome can be supported in reaction-function equilibrium, including miscoordination; 
see Theorem~\ref{TH:two-maxmin}.
Whereas the profile of best replies, $\textit{BR}$, is not unambiguous and thus not a reaction-function equilibrium, the \emph{worst} replies actually support miscoordination, that is $(x,y)$ or $(y,x)$.
This is an example of a ``bad equilibrium'' in which each player rationalizes playing badly by that the other does so as well.
This issue is inherited directly from Nash equilibrium (e.g. the Pareto-inferior Nash equilibrium in the \emph{Divide the dollar}-game, see \citealp{BramsTaylor1994}) but gets amplified as reaction-function games are more complex.
}

\begin{observation} \label{OBS:constant}
    A profile of constant reaction functions is a reaction-function equilibrium if and only if its unique fixed point is a Nash equilibrium in the strategic game.
\end{observation}

To conclude whether a profile is a reaction-function equilibrium or not, it suffices to consider constant deviations.
This is recorded in Observation~\ref{OBS:constdeviation}.
Indeed, if player $i$ targets the preferred fixed point $a' \in \mathcal{E}(R'_i, R_{-i})$ by deviating through~$R'_i$, then $i$ also beneficially deviates through the constant reaction $R''_i$ comprised of always playing~$a'_i$ as then $a' \in \mathcal{E}(R''_i,R_{-i})$.%
\footnote{This parallels mixed-strategy equilibrium, where it analogously suffices to consider pure-strategy deviations.} 

\begin{observation} \label{OBS:constdeviation}
    A profile $R\in {\cal R}$ is a reaction-function equilibrium if and only if there does not exist a player $i$ and constant reaction function $R'_i \in \mathcal{R}_i$ such that $U_i(R'_i, R_{-i}) > U_i(R)$.
\end{observation}

Best-reply reactions are a natural candidate for equilibrium play.
However, in reaction-function games, the equilibrium condition is evaluated in terms of the fixed points induced by the profile, and a unilateral deviation can completely change this set.
Hence, even if, for instance, the profile of best replies has a unique fixed point, it need not be stable:
a player may prefer to deviate (through a constant reaction function)---possibly even corresponding to an action \emph{dominated} in the strategic game.
Example~\ref{EX:BR} illustrates this mechanism.

\begin{example}[Deviating from best-reply reactions] \label{EX:BR}
    Consider the two-player game with payoffs in Table~\ref{TAB:BR}.
    Best replies are highlighted:
    we have $\textit{BR}_1(x) = a$ and $\textit{BR}_1(y) = b$ whereas action $x$ dominates $y$ for the second player, $\textit{BR}_2(a) = \textit{BR}_2(b) = x$.
    
    \begin{table}[!htb]
        \renewcommand{\arraystretch}{1.2}
        \centering
        $\begin{array}{r | cc}
              &  x  &  y  \\ \hline
            a & \underline{1},\underline{1} & 0,0 \\
            b & 0,\underline{3} & \underline{1},2 
        \end{array}$
        \caption{Payoff matrix for Example~\ref{EX:BR}.
        Best replies are highlighted.}
        \label{TAB:BR}
    \end{table}

    There is a unique fixed point, $\mathcal{E}(\textit{BR}) = \{(a,x)\}$.
    However, $\textit{BR}$ is not a reaction-function equilibrium as player $2$ prefers to deviate through the constant reaction $R'_2$ with $R'_2(a) = R'_2(b) = y$:
    we have $\mathcal{E}(\textit{BR}_1, R'_2) = \{(b,y)\}$, so $U_2(\textit{BR}_1,R'_2) = 2 > 1 = U_2(\textit{BR})$.
    Indeed, $(\textit{BR}_1,R'_2)$ is a reaction-function equilibrium.
    \hfill $\circ$
\end{example}

There are often many reaction-function equilibria supporting the same outcome.
This mirrors how a wide range of off-equilibrium threats and promises can support the same equilibrium path in infinitely repeated games.
It will at times be more interesting to explore which outcomes can be supported rather than the particular reaction functions used to support them.
Formally, outcome $a$ is \textbf{supported} in reaction-function equilibrium if $a$ is the outcome unambiguously associated to an equilibrium profile;
that is, there is a reaction-function equilibrium $R \in \mathcal{R}$ such that $a \in \mathcal{E}(R)$ and $u_i(a) \geq u_i(a')$ for each player $i$ and outcome $a' \in \mathcal{E}(R)$.

\section{What outcomes can be supported?} \label{SEC:supportedOutcomes}

In this section, we first establish existence of reaction-function equilibria (Theorem~\ref{TH:existence}).
We then turn to two-player games, which make out a cornerstone of the game-theoretic literature.
We find a simple characterization of all two-player outcomes that can be supported in reaction-function equilibrium (Theorem~\ref{TH:two-maxmin}).
The proof is constructive and illustrates how reaction functions can emulate strategies familiar from infinitely repeated games to support outcomes not attainable in one-shot strategic games.
With yet more players, equilibrium multiplicity becomes an issue as essentially all outcomes can be supported (Theorem~\ref{TH:moreplayers}).

\subsection{Equilibrium existence}

Our next result establishes that reaction-function equilibria always exist.
This is in contrast to Nash equilibrium in the corresponding strategic game (recall that players cannot mix).
By Observation~\ref{OBS:constant}, this implies that we will have to go beyond constant reaction functions to prove reaction-function equilibrium existence.
We formalize the argument below, which is instructive and of independent interest as it links to yet another class of games, namely sequential-move games.
These neatly bridge simultaneous-move strategic games and reaction-function games.

Suppose players take action in sequence $1, \dots, n$, each observing all prior choices $a_{<i}\in A_{<i}$ before making theirs (in the conventional sense, i.e., not using reaction functions).
This defines a finite extensive-form game of perfect information.
Such games do have Nash equilibria \citep[use backward induction]{Kuhn1953}.
Such an equilibrium is a strategy profile $s = (s_1, \dots, s_n)$ where each $s_i \colon A_{<i} \to A_i$ maps histories of prior moves $(a_1, \dots, a_{i-1}) \in A_{<i}$ to an action $s_i(a_1, \dots, a_{i-1}) \in A_i$ by player $i$.
This translates into the reaction function $R_i \in \mathcal{R}_i$ with $R_i(a_{-i}) = s_i(a_{<i})$ for all $a_{-i} \in A_{-i}$.
As sketched in Remark~\ref{REM:uniqueness}, there is a unique fixed point at~$R$, namely the outcome realized along the equilibrium path of the extensive-form game, and the same holds at $(R'_i,R_{-i})$ for all constant reactions $R'_i \in \mathcal{R}_i$. 
If $i$ were able to deviate in the reaction-function game by always taking action~$a'_i$, then $i$ would also deviate in the extensive-form game through strategy $s'_i$ with $s'_i(a_{<i}) = a'_i$ for each~$a_{<i} \in A_{<i}$.
Therefore, outcomes supported in Nash equilibrium in the extensive-form game---which exist---are also supported in reaction-function equilibrium.\footnote{%
One may wonder whether the converse is true, that is, if outcome $a$ is supported in reaction-function equilibrium, can we order the players such that $a$ is realized on the equilibrium path of the corresponding extensive-form game?
This is not true.
The ``bad equilibria'' of the \textit{Battle of the sexes} (see footnote~\ref{FN:BoS}) cannot be supported in this way as they require a ``cyclical'' reaction-function profile (compare Remark~\ref{REM:uniqueness}).}

\begin{theorem} \label{TH:existence}
    In every reaction-function game, there is a reaction-function equilibrium.
\end{theorem}

If we apply this construction to \textit{Matching pennies} (Table~\ref{TAB:MP}), we find profiles that comprise one constant and one best-reply reaction (e.g., $R_1(H) = R_1(T) = H$ and $R_2 = \textit{BR}_2$ with $R_2(H) = T$ and $R_2(T) = H$).
These are indeed the only reaction-function equilibria.

\begin{table}[!htb]
    \renewcommand{\arraystretch}{1.2}
    \centering
    $\begin{array}{r | cc}
          &  H  &  T  \\ \hline
        H & \phantom{-}\underline{1},-1 & \underline{-1},\phantom{-}\underline{1} \\
        T & -1,\phantom{-}\underline{1} & \phantom{-}1,-1
    \end{array}$
    \caption{Payoff matrix for \emph{Matching pennies}.
    A reaction-function equilibrium supporting $(H,T)$ is highlighted.}
    \label{TAB:MP}
\end{table}

\subsection{Maxmin payoff bound for two-player games}

Our next result mimics the ``folk theorems'' obtained in \citet{Tennenholtz2004} and \citet{Kalaietal2010} for two-player games. 
Theorem~\ref{TH:two-maxmin} asserts that an outcome can be supported precisely when it awards every player at least their \emph{maxmin payoff} \citep[or ``security level'',][]{Maschler_Solan_Zamir_2013}. Define player $i$'s \textbf{maxmin payoff} $\ubar{v}_i \in \mathbb{R}$ as the highest payoff $i$ can ensure irrespective of the actions of the others, i.e.,
\[
    \ubar{v}_i \equiv \max_{a_i \in A_i} \min_{a_{-i} \in A_{-i}} u_i(a_i, a_{-i}).
\]
We will let $\ubar{a}_i \in A_i$ denote a corresponding ``safe'' action, which generically will be unique.

Theorem~\ref{TH:two-maxmin} parallels folk theorems of infinitely repeated games but is now obtained in a one-shot setting.
The ``trigger strategy'' of repeated games here maps into a ``promise and threat'' reaction function.
To support outcome $a^* \in A$ (with at least maxmin payoffs), play $R_i(a^*_{-i}) = a^*_i$ and, for $a_{-i} \neq a^*_{-i}$, select $a_i = R_i(a_{-i})$ to minimize the other player's payoff $u_{-i}(a)$.
Compared to the existence proof of Theorem~\ref{TH:existence}, which involved at least one constant reaction, we now go yet further and both players choose more complex reactions.
The proof is in the Appendix.

\begin{theorem} \label{TH:two-maxmin}
    For two-player games, outcome $a \in A$ is supported in reaction-function equilibrium if and only if $u_i(a) \geq \ubar{v}_i$ for each player~$i$. 
\end{theorem}

To illustrate Theorem~\ref{TH:two-maxmin}, consider the \emph{Prisoner's dilemma} (Table~\ref{TAB:PD}).
In the strategic game, the Pareto-efficient cooperative outcome $(C,C)$ cannot be supported as ``defecting'' ($D$) is dominant for each player.
However, the cooperative outcome can be supported in the reaction-function game, and in precisely one way:
each player cooperates if and only if the other player does (highlighted in Table~\ref{TAB:PD}).
Such reactions lead to two Pareto-ordered fixed points, the inferior Nash equilibrium $(D,D)$ and the superior outcome $(C,C)$. 
Hence, compared to the unique fixed point in the construction offered in Theorem~\ref{TH:existence}, we may now obtain reaction-function equilibria with multiple (but still Pareto-ordered) fixed points.

\begin{table}[!htb]
    \renewcommand{\arraystretch}{1.2}
    \centering
    $\begin{array}{r | cc}
          &  C  &  D  \\ \hline
        C & \underline{2},\underline{2} & 0,3 \\
        D & 3,0 & \underline{1},\underline{1}
    \end{array}$
    \caption{Payoff matrix for the \emph{Prisoner's dilemma}.
    The unique reaction-function profile that supports $(C,C)$ is highlighted.}
    \label{TAB:PD}
\end{table}

\subsection{More players, more possibilities for support}

Once there are more players, we can carefully construct reactions to support an outcome by ruling out any other outcome from ever becoming a fixed point even following unilateral deviations.
That is, for each outcome $a$, we can find a profile $R \in \mathcal{R}$ such that, for each deviation $R'_i \in \mathcal{R}_i$, we have $\mathcal{E}(R'_i,R_{-i}) \subseteq \mathcal{E}(R) = \{ a \}$ and thus $U_i(R'_i,R_{-i}) \leq U_i(R)$.%
\footnote{Importantly, this uses the assumption $\mathcal{E}(R'_i,R_{-i}) = \emptyset \implies U_i(R'_i,R_{-i}) = - \infty$.}
This illustrates how an outcome can be supported %
simply because there is no other alternative.
The proof is a combinatorial exercise found in the Appendix.
For the case with three players and two actions, we can find examples in which some outcomes cannot be supported in reaction-function equilibrium (see again the Appendix).

\begin{theorem} \label{TH:moreplayers}
    All outcomes can be supported in reaction-function equilibrium if
    \begin{itemize}[noitemsep]
        \item There are at least four players, or
        \item There are at least three players with at least three actions each.
    \end{itemize}
\end{theorem}

The construction employed in the proof of Theorem~\ref{TH:moreplayers} is highly artificial and provides little in terms of practical insight---for instance, player reactions are set completely independent from their payoffs.
Up to this point, we have imposed little structure on reactions.
This helps to establish existence and broad support results, but also allows play that is hard to meaningfully interpret.
In what follows, we introduce additional constraints on reaction functions (starting with \emph{safe play}) to move from theoretical possibility to more compelling predictions and to reduce support for ``bad'' outcomes.

\section{``Playing it safe'' with reaction functions} \label{SEC:safe}

When constructing reaction functions, some reactions appear less safe than others.
For instance, Theorem~\ref{TH:two-maxmin} builds on ``promise and threat'' reactions, where an outcome with at least maxmin payoffs is supported on the premise of maximal punishment for any deviation.
Even though these ``threats'' are not carried out in equilibrium, it seems risky for player $i$ to base their reaction solely on player $j$'s payoffs.
Next, we address this by introducing the idea of \emph{safe play}.

``Safe play'' in strategic games entails taking a safe action~$\ubar{a}_i$ to ensure at least maxmin payoff~$\ubar{v}_i$ regardless the actions of the others \citep[e.g.][]{Maschler_Solan_Zamir_2013}.
This individual risk-minimizing behavior can lead to collectively bad outcomes, for instance, in \emph{Prisoner's dilemma} (Table~\ref{TAB:PD}).
With reaction functions, there are more options.
We say that reaction $R_i$ is \emph{safe} if, for all actions $a_{-i}$ of the others, the outcome $(R_i(a_{-i}), a_{-i}) $ awards $i$ at least her maxmin payoff~$\ubar{v}_i$.\footnote{%
A stronger requirement would be $U_i(R_i,R_{-i}) \geq \ubar{v}_i$ for each $R_{-i} \in \mathcal{R}_{-i}$, but this is too demanding.
With more than two players, there is $R_{-i} \in \mathcal{R}_{-i}$ such that $\mathcal{E}(R) = \emptyset$ and thus $U_i(R) = -\infty$ for all $R_i \in \mathcal{R}_i$.}
A reaction-function equilibrium $R \in {\cal R}$ is \emph{safe} if each player's reaction $R_i \in {\cal R}_i$ is safe. 

\begin{definition}
    Reaction function $R_i \in \mathcal{R}_i$ is \textbf{safe} if $u_i(R_i(a_{-i}), a_{-i}) \geq \ubar{v}_i$ for each $a_{-i} \in A_{-i}$.
\end{definition}

Safe reaction-function play is always possible, for instance by playing best replies $\textit{BR}_i$.
Yet simpler is the constant reaction function to always play a safe action $\ubar{a}_i$.
If the resulting action profile $\ubar{a}$ of such play is a Nash equilibrium in the strategic game, then the corresponding constant reactions form a reaction-function equilibrium (Observation~\ref{OBS:constant}).
In cases like a ``race to the bottom'', where players, say, will not exert effort unless others do, the low-effort outcome may be a Nash equilibrium and thus supported in safe reaction-function equilibrium.
But importantly, if Pareto improvements are possible, they too can be supported in safe reaction-function equilibrium.
Hence, whenever a safe reaction-function equilibrium exists, a Pareto-efficient one does as well.

\begin{theorem} \label{TH:safeRFE}
    If a safe outcome $\ubar{a} \in A$ of the strategic game is a Nash equilibrium, %
    then there is a safe reaction-function equilibrium.
    In general, if a safe reaction-function equilibrium exists, then there exists a Pareto-efficient safe reaction-function equilibrium.
\end{theorem}

To illustrate Theorem~\ref{TH:safeRFE} and how reaction functions make a difference compared to strategic games, consider the \emph{Stag hunt}  (tracing back to Rousseau's \textit{A Discourse on Inequality}) in Table~\ref{TAB:safeexist}.
In the strategic game, both taking safe action $y$ leads to the risk-dominant Nash equilibrium \citep{HarsanyiSelten1988}.
In contrast, reacting by matching the action of the other (highlighted) is safe and supports the superior payoff-dominant outcome $(x,x)$.
Thus, safe reaction-function equilibrium ensures a high-enough lower bound on player payoffs without having to compromise in achieving yet better outcomes.

\begin{table}[!htb]
    \renewcommand{\arraystretch}{1.2}
    \centering
    $\begin{array}{r | cc}
          &  x  &  y  \\ \hline
        x & \underline{3},\underline{3} & 0,2 \\
        y & 2,0 & \underline{2},\underline{2}
    \end{array}$
    \caption{Payoff matrix for \emph{Stag hunt}.
    The ``match other'' reactions are highlighted, supporting $(x,x)$ in safe reaction-function equilibrium.}
    \label{TAB:safeexist}
\end{table}

However, existence is not guaranteed.
Example~\ref{EX:BR} already showed that all playing best replies, i.e., playing safe, need not be an equilibrium.
Example~\ref{EX:nosafeRFE} goes further to show that there need not exist a safe equilibrium.

\begin{example}[No safe reaction-function equilibrium] \label{EX:nosafeRFE}
    Consider the two-player game with payoffs in Table~\ref{TAB:nosafeRFE}.
    For contradiction, suppose $R$ is a safe reaction-function equilibrium and label the supported outcome~$a$.
    Label the players $i$ and $j$ such that $u_i(a) < 3$.
    
    \begin{table}[!htb]
        \renewcommand{\arraystretch}{1.2}
        \centering
        $\begin{array}{r | ccc}
              &  x  &  y  &  z \\ \hline
            x & 1,1 & 1,0 & 1,3 \\
            y & 0,1 & 2,2 & 0,1 \\
            z & 3,1 & 1,0 & 0,0
        \end{array}$
        \caption{Payoff matrix for Example~\ref{EX:nosafeRFE}.}
        \label{TAB:nosafeRFE}
    \end{table}
    Each player's maxmin payoff is~$1$, so safe play dictates that $R_j(z) = x$.
    But then player $i$ deviates by always playing~$z$ as the unique fixed point yields payoff $u_i(z,x) = 3$.
    This contradicts that $R$ is a reaction-function equilibrium.
    \hfill $\circ$
\end{example}

Example~\ref{EX:nosafeRFE} gives further insights on our main results in Section~\ref{SEC:supportedOutcomes}.
In consequence, the construction for our existence result (Theorem~\ref{TH:existence}) involves unsafe play and our two-player characterization (Theorem~\ref{TH:two-maxmin}) does not extend to safe reactions.
It is by definition that ``bad'' outcomes awarding less than maxmin payoffs cannot be supported in safe play;
hence, safe play addresses at least the most problematic aspect of equilibrium multiplicity (Theorem~\ref{TH:moreplayers}).
Most reaction functions that we explore next, in Section~\ref{SEC:investmentGames}, will be safe.

\section{Symmetric investment games} \label{SEC:investmentGames}

We now apply our theoretical framework on some games for which implementing reaction functions through blockchain-based smart contracts seems very promising.
These are symmetric $n$-player games with finite action sets $A_i = \{ 0, 1, \dots, H \}$ interpreted as investment amounts.
Reaction functions now go from being a theoretical concept to a most practical tool.
Conditional investment strategies provide an excellent use case for smart contracts \citep[e.g.][]{WeberStaples2022}.
Payoffs for \emph{symmetric investment games} include a non-decreasing common value $v \colon A \to \mathbb{R}$ derived symmetrically from costly individual actions.%
\footnote{That is, for each outcome $a \in A$ and permutation $\pi \colon N \to N$, we have $v(a) = v(a_{\pi(1)}, \dots, a_{\pi(n)})$.}
Specifically, for each player $i$, and outcome $a \in A$,
\[
    u_i(a) = v(a) - a_i.
\]
As players are symmetric, \emph{coordinated} outcomes ($a_i = a_j$ for all players $i$ and~$j$) will be key. 
For convenience, we assume that these coordinated outcomes are Pareto ranked and that the higher players coordinate the better:
$u_i(\alpha, \dots, \alpha)$ is increasing in $\alpha \in \{0, \dots, H\}$.
Additional structure will be imposed on $v$ resulting in \emph{weakest-link} (Subsection \ref{SUB:WL}) and \emph{public-good} games (Subsection \ref{SUB:PG}), respectively.

In what follows, we focus on monotone reaction functions (i.e., players commit to invest more the more others invest).
Monotone reaction functions guarantee fixed-point existence and makes computation of such fixed points easy.
Let $\mathcal{M}_i \subseteq \mathcal{R}_i$ denote the set of player $i$'s monotone reaction functions.

As action spaces $A_i$ are identical and contribution to the total value~$v$ is symmetric across players, it is natural to consider \emph{symmetric} reaction functions.
In essence, the identities of the other players should not matter.
A reaction function $R_i \in \mathcal{R}_i$ is \textbf{symmetric} if, for each $a_{-i} \in A_{-i}$ and permutation $\pi \colon N \to N$ with $\pi(i) = i$, we have $R_i(a_{-i}) = R_i(a_{\pi(-i)})$ with $a_\pi \equiv (a_{\pi(1)}, \dots, a_{\pi(n)}) \in A$.
Let $\mathcal{S}_i \subseteq \mathcal{M}_i$ denote the set of player $i$'s symmetric and monotone reaction functions.%
\footnote{%
For instance, in oligopoly models, a firm's profit typically is unchanged if its competitors' actions (prices, quantities) are permuted.
In this case, a symmetric reaction function specifies the same reaction to two such permuted action profiles.}

With symmetric reaction functions, it becomes meaningful to compare functions across players.
We say that $i$ and $j$ play \emph{the same} reaction function, denoted $R_i \cong R_j$, whenever $R_i(x) = R_j(x)$ for each $x \in \{ 0, \dots, H \}^{n-1}$. That is, players $i$ and $j$ react in the same way to the actions of others.
Let $\mathcal{N} \subset \mathcal{S} \subset \mathcal{R}$ be the set of profiles in which \emph{all} players react in the same way, that is, $R_i \cong R_j$ for all players $i$ and $j$.
We think of such profiles as \textit{social norms} or conventions, and say that a reaction function $R^*_i$ is \emph{norm-proof} if, for every norm~$R \in \mathcal{N}$, playing $R^*_i$ against the norm $R$ yields at least as high payoff as conforming to it.%
\footnote{Technically, the property bears resemblance to the ``evolutionary stable strategy'' (ESS) refinement of Nash equilibrium prominent in evolutionary game theory \citep[][]{MaynardSmithPrice1973}. Further evidence to this point is given in the Appendix, where we simulate ``evolution'' of a pool of reaction functions in a public-good setting.
The simulations indicate that evolutionary selection indeed would push towards a pool of mainly norm-proof reactions.}
\begin{definition}
    Reaction function $R^*_i \in \mathcal{M}_i$ is \textbf{norm-proof} if, for each norm $R \in \mathcal{N}$,
    \[
        U_i(R^*_i, R_{-i}) \geq U_i(R).
    \]
\end{definition}

This is not an equilibrium property, in the sense that we do not require $R^*_i$ to be optimal against a specific $R_{-i}$, but rather an alternative interpretation of safe play.
Even if you do not know which norm $R$ the others might settle on, playing $R^*_i$ is a safe and reasonable choice at least as good as conforming to the norm.

Norm-proof reaction functions turn out to have a certain structure in symmetric investment games. 
We say that a reaction function $R_i$ exhibits \emph{conditional collaboration} if it satisfies the following two criteria:
if the others coordinate (play the same action), then $i$ does as well; 
if not, then $i$ matches at least the minimum investment and at most the maximum investment.
That is, $R_i \in \mathcal{R}_i$ exhibits \textbf{conditional collaboration} if, for each $\alpha \in \{ 0, \dots, H\}$, $R_i(\alpha,\dots,\alpha) = \alpha$ and, for all remaining $a_{-i} \in A_{-i}$, 
\[
    \textstyle \min_{j \neq i} a_j \leq R_i(a_{-i}) \leq \max_{j \neq i} a_j.
\]

The principle of conditional collaboration is familiar from behavioral economics, where experimental evidence has shown that participants are often willing to contribute more the more others contribute \citep[e.g.][]{Fischbacheretal2001}.
This has typically been rationalized by preferences reflecting altruism or reciprocity \citep[e.g.][]{Sugden1984,Croson2007}.
We offer an alternative perspective where conditional collaboration arises even without such ``moral preferences''.
For reaction functions, conditional collaboration can  be rationalized through norm-proofness.

\begin{observation} \label{OBS:normproof}
    For symmetric investment games, a monotone reaction function is norm-proof if and only if it exhibits conditional collaboration.
\end{observation}

There is a unique reaction function exhibiting conditional collaboration if there are only two players (``match the other's action''), %
which is safe.
However, with more players, the class is large and may contain unsafe reactions.

Next, we present a novel consistency property to capture that $i$'s reaction foremost should be determined by $i$'s own payoffs.
The reaction should therefore be consistent across \emph{payoff-equivalent} action profiles.
Formally, if $i$'s payoffs are the same at $a_{-i}$ and $a'_{-i}$ regardless of the action that $i$ takes, so $u_i(a_i,a_{-i}) = u_i(a_i,a'_{-i})$ for each action~$a_i$, then $i$'s reaction should also be the same, $R_i(a_{-i}) = R_i(a'_{-i})$.
This is a refinement of the symmetry property.

\begin{definition}
    Reaction function $R_i \in \mathcal{R}_i$ is \textbf{payoff consistent} if, for each $\{ a_{-i}, a'_{-i} \} \subseteq A_{-i}$,
    \[
        \left ( u_i(a_i, a_{-i}) = u_i(a_i, a'_{-i}) \text{ for each } a_i \in A_i \right ) 
        \implies R_i(a_{-i}) = R_i(a'_{-i}).
    \]
\end{definition}

It is immediate that payoff-consistent reaction functions exist (e.g., constant reactions).
The best-reply reaction (with a fixed tie-breaking rule if needed) is also payoff consistent.
For many games, the property will have no bite whatsoever as there might not be any payoff-equivalent action profiles.
However, in the games that we explore next in Subsections~\ref{SUB:WL} and~\ref{SUB:PG}, payoff consistency combines with norm-proofness to give a unique recommendation.

\subsection{Weakest-link games} \label{SUB:WL}

In weakest-link (or ``minimum-effort'') games, %
the minimum investment determines the players' common value $v(a)$; the higher this is, the better for all \citep[e.g.][]{Hirshleifer1983}.
These games are abundant in practice and arise, for instance, when divisions within a firm contribute independently to a joint project for which the lowest-quality input disproportionally affects the output \citep{KnezCamerer1994,Kunreuther2009,GudmundssonJEBO2022}. Specifically, payoffs are given by 
\[
    u_i(a) = \lambda \cdot \textstyle \min_j a_j - a_i
\]
with parameter $\lambda > 1$.
Hence, players wish to match the minimum investment of the others and ideally do so at~$H$.
Conventional safe play amounts to taking action $a_i = 0$ with $\ubar{v}_i = u_i(0,\dots,0) = 0$.
Moreover, action profiles $a_{-i}$ and $a'_{-i}$ are payoff equivalent whenever $\min_{j \neq i} a_j = \min_{j \neq i} a'_j$.

The best-reply reaction $\textit{BR}_i$ is particularly simple for weakest-link games and entails matching the minimum investment, $\textit{BR}_i(a_{-i}) = \min_{j \neq i} a_j$.
Not only is it safe, it is also the only norm-proof and payoff-consistent reaction function.%
\footnote{In fact, $\textit{BR}_i$ is the only norm-proof reaction that is safe for all parameter values $\lambda > 1$.}

\begin{theorem} \label{TH:WLchar}
    For weakest-link games, reaction function $R_i \in \mathcal{R}_i$ is norm-proof and payoff consistent if and only if $R_i = \textit{BR}_i$.
\end{theorem}

Playing $BR_i$ against a profile of reaction functions that exhibit conditional collaboration always supports the high-investment outcome $(H,\dots, H)$ in reaction-function equilibrium.

In weakest-link games, there are compelling reasons to expect that players use monotone reaction functions, in which case we obtain an additional argument in favor of $\textit{BR}_i$. 
For this, we restrict attention to what we call \emph{high-risk} weakest-link games with parameter $\lambda$ such that $1 < \lambda < H / (H-1)$.
That is, coordination gains are small relative to effort costs to the point that coordinating at $0$ is superior to playing $H$ against $H-1$.
This does not mean that there is no gain from coordinating at a higher investment level---the payoff is still higher at $(H,\dots, H)$ than at $(0,\dots,0)$---rather, it means that higher investments come with greater risk.%
\footnote{%
For comparison, when \citet{vanHuycketal1990} found that experiment participants failed to coordinate, the setting corresponded to $H = 6$ and $\lambda = 2$.
This is a much easier environment to coordinate in than high-risk games, where we would have $\lambda < H/(H-1) = 1.2$.}

We define reaction function $R_i \in \mathcal{R}_i$ to be \textbf{weakly dominant against monotone play} if, for each $R'_i \in \mathcal{R}_i$ and $R_{-i} \in \mathcal{M}_{-i}$, $U_i(R) \geq U_i(R'_i,R_{-i})$.
That is, the player cannot do better than $R_i$ against monotone reactions $R_{-i} \in \mathcal{M}_{-i}$.
Note that the statement does not restrict $i$'s own reaction ($R_i$ and $R'_i$, respectively) to be monotone.
However, the characterization in Theorem~\ref{TH:WLmonotone} yields a unique reaction function, namely $\textit{BR}_i$, which is monotone.

\begin{theorem} \label{TH:WLmonotone}
    For high-risk weakest-link games, reaction function $R_i \in \mathcal{R}_i$ is weakly dominant against monotone play if and only if $R_i = \textit{BR}_i$.
\end{theorem}

The high-risk qualification is restrictive, but this is to be expected.
There are many reaction functions that ``solve'' weakest-link games in general.
The purpose of Theorem~\ref{TH:WLmonotone} is partly to identify a strong-enough condition under which we can make a unique, precise recommendation of which reaction function to play.

There are other variations on the same idea such as weak dominance against safe play, where one assumes instead the others to play safe.
In some sense, safe play is more restrictive than monotone play.
Therefore, there are now many weakly dominant reaction functions, for instance all that exhibit conditional collaboration (including $\textit{BR}_i$).
One can even define weak dominance in general (so not only against safe or monotone play), but we argue in Remark~\ref{REM:WD} that this is too demanding. 

\begin{remark}[Weak dominance in general] \label{REM:WD}
    Consider a two-player weakest-link game.
    Suppose, for contradiction, that there exists a reaction function $\hat{R}_i \in \mathcal{R}_i$ that is weakly dominant in general (so not only against monotone play, for instance).
    Suppose first the other player chooses $R_j$, which is to always play~$H$.
    Then we must have $\hat{R}_i(H) = H$ to ensure that the optimal outcome $(H,H)$ is a fixed point.
    Consider next the case in which the other player chooses reaction $R'_j$ such that $R'_j(H-1) = H$ and otherwise $R'_j(a_i) = 0$, which is neither monotone nor safe (in a high-risk setting). 
    The best attainable fixed point now is $(a_i,a_j) = (H-1,H)$, but this would require $\hat{R}_i(H) = H-1$, which is a contradiction.
    Hence, no reaction function is optimal against both $R_j$ and~$R'_j$.
    \hfill $\circ$
\end{remark}

The fundamental message to take from this section is that reaction functions offer a new way to overcome trust issues in weakest-link games.
Reaction-function play does not rely on building trust between the players \citep{ChenChen2011}, on putting hope to communication routines \citep{BlumeOrtmann2007,GudmundssonJEBO2022}, or on establishing a central authority to monitor and punish inefficient action \citep{Holmstrom1982}.
None of this is needed---using reaction functions as commitment devices would in principle allow any large group of strangers to coordinate on their very first attempt.

Once we dig deeper and identify which particular reaction function to play, it is no surprise that the best-reply reaction function stands out (i.e., to match the weakest link).
In a sense, Theorems~\ref{TH:WLchar} and~\ref{TH:WLmonotone} together with the surrounding discussion make for a robustness exercise that verifies that, indeed, the intuitive best-reply reaction is the ideal reaction function to play in this case.

\subsection{Public-good games} \label{SUB:PG}

We turn to a different class of games for which, in contrast, best replies are ineffective.
We have already seen that reaction functions introduce new ways for players to coordinate around Pareto-improving outcomes in strictly dominated actions (e.g. in the \emph{Prisoner's dilemma}).
Now, we will further expand on this in the particular context of public-good games, in which each player's investment benefits everybody, but the investment is a net loss to the player herself \citep[e.g.][]{BergstromEtAl1986}.
In public-good games, payoffs are given by
\[
    u_i(a) = \lambda \cdot \textstyle \sum_j a_j - a_i
\]
with parameter $1/n < \lambda < 1$.%
\footnote{Outside this parameter range, high investment (for $\lambda \geq 1$) or zero investment (for $\lambda \leq 1/n$) is both safe, dominant, and efficient.}
High investment $(H, \dots, H)$ maximizes total payoffs $\sum_j u_j(a)$ but, in the strategic game, free-riding through~$a_i = 0$ is dominant.
Hence, the best reply is simply to free-ride and never invest.
Again, conventional safe play yields $\ubar{v}_i = u_i(0,\dots,0) = 0$.
Now, action profiles $a_{-i}$ and $a'_{-i}$ are payoff equivalent if they add up to the same:
$\sum_{j \neq i} a_j = \sum_{j \neq i} a'_j$.

Whereas in the weakest-link games the key idea was to react to the minimum action of the others, in the public-good setting it is instead the \emph{average} contribution that turns out to be essential.
In particular, norm-proof and payoff-consistent reaction functions must match the average, either rounded up or down.%
\footnote{For $z \in \mathbb{R}$, let $\lfloor z \rfloor, \lceil z \rceil \in \mathbb{Z}$ be the largest (smallest) integer no larger (smaller) than~$z$.}

\begin{theorem} \label{TH:PGchar}
    For public-good games, reaction function $R_i \in \mathcal{M}_i$ is norm-proof and payoff consistent if and only if there is a monotone function $f_i \colon \mathbb{R} \to \mathbb{Z}$ such that, %
    for each $a_{-i} \in A_{-i}$,
    \[
        R_i(a_{-i}) 
        = f_i \left (\frac{1}{n-1} \textstyle \sum_{j \neq i} a_j \right )
        \in \left \{
            {\left \lfloor \frac{1}{n-1} \textstyle \sum_{j \neq i} a_j \right \rfloor}, 
            {\left \lceil \frac{1}{n-1} \textstyle \sum_{j \neq i} a_j \right \rceil}
            \right \}.
    \]
\end{theorem}

Theorem~\ref{TH:PGchar} leaves two particularly focal candidates within the class of monotone, norm-proof, and payoff-consistent reactions:
to match the average, either consistently rounding up or consistently rounding down.
Of these, we will argue that rounding down (to be denoted $R^*_i$ below) fares better.
Whereas rounding up is a more ``generous'' commitment, it does not effectively discipline free-riding, much like overly cooperative strategies do poorly in a repeated \textit{Prisoner's dilemma} \citep[compare][]{Axelrod1984}.
If all use reactions that round up, players would beneficially deviate by undercutting (for instance, by playing the constant reaction associated to investment~$H-1$).
For contrast, rounding down supports high investment in reaction-function equilibrium.
On top of this, rounding down is safe across all parameter values, whereas ``rounding up'' is not for certain parameter values.

Beyond our equilibrium and refinement analysis, it is useful to benchmark reactions against a simple normative criterion:
choose the investment that maximizes total welfare, subject to individual safe play.
This also connects to the the experimental literature on social preferences \citep[e.g.][]{CharnessRabin2002}, where players sometimes behave as if they partially internalize others' payoffs.
Specifically, we define a reaction function $R_i \in \mathcal{R}_i$ to be \textbf{welfare-maximizing conditional on safe play} if, for each $a_{-i} \in A_{-i}$,  
\[
    R_i(a_{-i}) \in \textstyle \arg \left ( \max_{a_i \in A_i} \sum_j u_j(a) \quad \text{subject to} \quad u_i(a) \geq \ubar{v}_i \right ).
\]
As a parallel to Theorem~\ref{TH:WLmonotone}, we define \emph{high-risk} public-good games with parameters $\lambda$ such that $1/n < \lambda < H / (nH-1)$.
The intuition is the same:
there are still gains from coordinating on higher investment, but the returns are small compared to the risks involved in overinvesting relative to the others.
For these games, matching the average, rounded down, is welfare-maximizing subject to safe play.%
\footnote{For weakest-link games, this property instead singles out $\textit{BR}_i$.
There is no reason for $i$ to exceed the minimum investment as this creates no added value to the others but only adds to $i$'s investment costs;
it is also not beneficial for $i$ to undercut the minimum, because this harms all players more than it saves on $i$'s costs.}

\begin{theorem} \label{TH:PG}
    For high-risk public-good games, $R_i \in {\cal M}_i$ is welfare-maximizing conditional on safe play if and only if  \[
        R_i^*(a_{-i}) = \left \lfloor \frac{1}{n-1} \textstyle \sum_{j \neq i} a_j \right \rfloor.
    \]
\end{theorem}

In this way, reaction functions provide new ways to overcome tendencies to free-ride in public-goods games, which do not rely on, for instance, players forming sanctioning institutions \citep{KosfeldEtAl2009}, being altruistic \citep{AndersonEtAl1998}, or the introduction of obligations \citep{Galbiati2008}.
Our main candidate $R^*_i$ is to match the average, rounded down, which appears to be a compelling rule-of-thumb for reaction-function play in public-good games.
We do not claim that this reaction function always is the optimal strategic choice (e.g., it is better to free-ride if the others unconditionally contribute), but it appears always quite reasonable:
it is simple, safe, and ensures that the player is no worse off than the average.
If faced with a pool of free-riders, it specifies to free-ride as well; 
against both conditional and unconditional collaborators, high investment is a fixed point. 
In this way, it does reasonably well against most reaction functions.
This is further corroborated by the simulation study in the Appendix.
There, in a specific public-good setting, we find compelling evidence that a process of evolutionary selection would converge towards functions alike~$R^*_i$.

\section{Concluding remarks} \label{SEC:conclusions}

Blockchain-based smart contracts make contingent commitments both credible and executable.
In our framework, players commit to reaction functions and a coordinating contract implements a fixed point of the resulting profile.
This is particularly well suited for investment environments, where players can escrow deposits and condition their investment on others' investments without risking over-exposure.
In such settings, reaction functions can help players effectively implement coordination that resolves issues of trust (as in weakest-link games) and disciplines free-riding (as in public-good games).
We posit some additional requirements (safe play, norm-proofness, and payoff consistency), which jointly yield simple and practically implementable recommendations.

Compared to the approaches in \citet{Tennenholtz2004} and \citet{Kalaietal2010}, we submit that the reaction functions bring automated play a step closer to practical application.
For example, to obtain collaboration in program equilibrium of the \emph{Prisoner's dilemma}, \citet{Tennenholtz2004} suggests a program of the form ``if my program is identical to my opponent's program, then the program collaborates; otherwise, it defects''.
Conditioning your action on the program of the opponent is fragile, as the smallest difference in programs yields defection.
Even though players may be willing to collaborate, programs can be formulated in countless ways and there is a risk that a collaborative outcome would never prevail in practice.
For instance, imagine that the other player chooses a different but cooperative-looking program (e.g., a program that unconditionally always collaborates).
Even though both are collaborative, the program mismatch still leads to defection.
In contrast, our reaction functions specify an action for each action profile of the others.
This may be more robust, for instance as illustrated in the cases of weakest-link and public-good games.
For example, if players exhibit conditional collaboration in a weakest-link game, we will not automatically resort to the worst-case outcome if some players deviate from the best-reply reaction function. 

We end on a few directions for future research.
One is to explore reaction functions in games with incomplete information.
For instance, the idea of ``playing safe'' seems appealing even if, say, you do not know the preferences of the other players.
A first step in that direction is Theorem~\ref{TH:WLmonotone}, which shows that the best-reply reaction is weakly dominant against monotone play in high-risk weakest-link games. 
As weak dominance is a statement only on the player's own payoffs, knowledge of the others' payoffs is not important.
One can ask more generally, in what games are best-reply reactions $\textit{BR}_i$ weakly dominant against monotone play?
Theorem~\ref{TH:WLmonotone} then points to some conditions:
(a) payoff $u_i(a)$ is non-decreasing in $a_{-i} \in A_{-i}$, (b) $\textit{BR}_i$ is monotone (to ensure fixed points  against monotone play), and (c), for each outcome $a \in A$,
\[
    a_i > \textit{BR}_i(a_{-i}) \implies u_i(a) \leq u_i(0,\dots,0).
\]
Under these assumptions, $\textit{BR}_i$ is weakly dominant against monotone play.

There are many games, different from high-risk weakest-link games, for which this applies.
The common theme is that (a) $i$ is not harmed by $j$'s actions, that (b) there are some complementarities between player contributions, and (c) that there is a ``target'' that the player wants to match but not exceed.
This target could for instance be derived from Kantian fairness reasoning \citep[e.g.][]{Roemer2010,AlgerWeibull2013} combined with a strong disutility from being ``exploited'' (investing more than what is ``right'' given the investments of the others to the others' benefit) or some form of inequality aversion \citep{FehrSchmidt}.
Given the prominent position of the best-reply reaction function, which is simple and intuitive, it would be interesting to map out exactly when it is also optimal.

\bibliography{bibliography}
\bibliographystyle{abbrvnat}

\section*{Appendix: Proofs} \label{APP:proofs}

\setcounter{section}{3}
\setcounter{theorem}{1}
\begin{theorem} 
    For two-player games, outcome $a \in A$ is supported in reaction-function equilibrium if and only if $u_i(a) \geq \ubar{v}_i$ for each player~$i$.
\end{theorem}

\begin{proof}
    \textsc{Part I:}
    Let $a^* \in A$ be such that, for each player $i$, $u_i(a^*) \geq \ubar{v}_i$.
    For each player $i$, define $R_i \in \mathcal{R}_i$ such that $R_i(a^*_{-i}) = a^*_i$, so $a^* \in \mathcal{E}(R)$, and otherwise $R_i(a_{-i})$ minimizes the other player's payoff:
    \[
        R_i(a_{-i}) \in \arg \min_{a_i \in A_i} u_{-i}(a_i, a_{-i}).
    \]
    
    Consider a potential deviation $R'_i \in \mathcal{R}_i$.
    If $\mathcal{E}(R'_i,R_{-i}) = \emptyset$, then $U_i(R'_i,R_{-i}) = - \infty < u_i(a^*)$ by convention, so the deviation is not beneficial.
    For each $a' \in \mathcal{E}(R'_i,R_{-i})$ that differs from $a^*$, $a'_{-i}$ minimizes $u_i$ given $a'_i$.
    That is,
    \[
        u_i(a') 
        = \min_{a_{-i} \in A_{-i}} u_i(a'_i,a_{-i})
        \leq \max_{a_i \in A_i} \min_{a_{-i} \in A_{-i}} u_i(a_i,a_{-i})
        = \ubar{v}_i
        \leq u_i(a^*).
    \]
    Hence, there is no beneficial deviation, so $R$ supports $a^*$ in reaction-function equilibrium.

    \bigskip \noindent
    \textsc{Part II:}
    For contradiction, say $\tilde{R} \in \mathcal{R}$ supports $\tilde{a} \in A$ in reaction-function equilibrium yet $\ubar{v}_i > u_i(\tilde{a})$ for some player~$i$. 
    Let $a'_i \in A_i$ be $i$'s maxmin strategy,
    \[
        a'_i \in \arg \max_{a_i \in A_i} \min_{a_{-i} \in A_{-i}} u_i(a_i,a_{-i}),
    \]
    and define $a'_{-i} = \tilde{R}_{-i}(a'_i)$.
    Let $R'_i \in \mathcal{R}_i$ be such that $R'_i(a_{-i}) = a'_i$ everywhere.
    Then $\mathcal{E}(R'_i,\tilde{R}_{-i}) = \{ a' \}$.
    Moreover, 
    \[
        u_i(a') 
        \geq \max_{a_i \in A_i} \min_{a_{-i} \in A_{-i}} u_i(a_i,a_{-i}) 
        = \ubar{v}_i 
        > u_i(\tilde{a}).
    \]
    Hence, $i$ deviates through~$R'_i$, contradicting that $\tilde{R}$ is a reaction-function equilibrium.
\end{proof}

\begin{theorem} 
    All outcomes can be supported in reaction-function equilibrium if
    \begin{itemize}[noitemsep]
        \item There are at least four players, or
        \item There are at least three players with at least three actions each.
    \end{itemize}
\end{theorem}

\begin{proof}
    \noindent
    \textsc{Part I:}
    Suppose first there are at least four players.
    We label actions such that $\{0,1\} \subseteq A_i$ for each player~$i$.
    For the purpose of illustration, we wish to support the outcome $(0, \dots, 0)$.
    Define reaction functions $R_i \in \mathcal{R}_i$ as follows:
    \[
        R_i(a_{-i}) = \begin{cases}
            1 & \text{ if } a_{i-1} = 0 \text{ and } a_{i+1} \neq 0 \text{ (mod $n$)} \\
            0 & \text{ otherwise.}
        \end{cases}
    \]
    That is, $i$ is influenced only by their ``neighbors'' and takes action $1$ if the player before them chooses $0$ whereas the one after them doesn't.
    Figure~\ref{FIG:cycle} illustrates this case.
    
    \begin{figure}[!htb]
        \centering
        \begin{tikzpicture}
            \def\n{4};
            \foreach \x [evaluate=\x as \deg using 90-\x*(360/(\n+1))] in {0,...,\n}{
                \coordinate (\x) at (\deg:1cm);
                \fill (\x) circle (2pt);
            }
            \node at (0) [anchor=south] {\small $i$};
            \node [black!50] at (0) [anchor=north] {\small $1$};
            \node at (1) [anchor=south west] {\small $i+1$};
            \node [black!50] at (1) [anchor=east] {\small $1$};
            \node at (\n) [anchor=south east] {\small $i-1$};
            \node [black!50] at (\n) [anchor=west] {\small $0$};
            \foreach \x [evaluate=\x as \y using {mod(\x+1,(\n+1))}] in {0,...,\n}{
                \path[bend left, ->, shorten <= .3em, shorten >= .3em] (\x) edge (\y);
            }
        \end{tikzpicture}
        \caption{Caption}
        \label{FIG:cycle}
    \end{figure}
    
    As $R_i(0,\dots,0) = 0$, the outcome $(0, \dots, 0)$ is a fixed point.
    Say player~$i$ is \emph{unhappy} at outcome $a$ if $R_i(a_{-i}) \neq a_i$. 
    We will argue that there are at least two unhappy players at each $a \neq (0,\dots,0)$. 
    For instance, $R_i(1,\dots,1) = 0$, so all players are unhappy at $(1,\dots,1)$.
    We restrict to the most challenging cases, namely $a \in \{0,1\}^n$.
    It remains therefore to check outcomes $a$ that include at least one $0$ and at least one~$1$.
    
    Suppose for contradiction that there is at most one unhappy player at~$a$.
    By construction, $a$ contains at least one sequence of consecutive $1$'s (possibly a single~$1$) followed by a~$0$.
    The last player $m$ of the sequence is unhappy as $a_m = 1$ and $a_{m+1} = 0$.
    Therefore, $a$ can contain at most one such sequence (else there would be at least two unhappy players).
    Moreover, each player $i$ not the first of the sequence is unhappy as $a_{i-1} = a_i = 1$.
    Hence, the sequence contains at most two players.
    There are two options:
    if $a_i = 1$ and $a_j = 0$ otherwise, then $i-1$ and $i$ are unhappy;
    if $a_i = a_{i+1} = 1$ and $s_j = 0$ otherwise, then $i-1$ and $i+1$ are unhappy.
    (In the second case, that $i-1$ is unhappy hinges on $a_{i-2} = 0$.
    As then $a_{i-2} \neq a_{i+1}$, this uses that $n \geq 4$.)
    Hence, at all outcomes $a \neq (0,\dots,0)$ there are at least two unhappy players.

    Consider now a unilateral deviation $R'_i \in \mathcal{R}_i$ by player~$i$ at profile~$R$.
    For contradiction, suppose some $a \neq (0,\dots,0)$ is a fixed point at $(R'_i,R_{-i})$.
    Because there are at least two unhappy players at $a$ under~$R$, there remains at least one unhappy player at $a$ under $(R'_i,R_{-i})$.
    Hence, $a$ is not a fixed point.
    Therefore, $\mathcal{E}(R'_i,R_{-i}) \subseteq \{ (0, \dots, 0) \}$ for all unilateral deviations, so $R$ supports $(0, \dots, 0)$ in reaction-function equilibrium.
    This argument is done without any reference to payoffs and the outcome $(0,\dots,0)$ is chosen arbitrarily.
    Any other outcome $a$ can be supported in the same way by appropriate relabeling.

    \bigskip \noindent
    \textsc{Part II:}
    The approach partly mirrors that of \textsc{Part~I}.
    Suppose instead that there are three players $N = \{ 0,1,2 \}$ and at least three actions $\{0,1,2\} \subseteq A_i$. 
    Construct reactions $R_i$ to support outcome $(0,0,0)$ as follows.
    Set $R_i(0,0) = 0$ and otherwise, for each $a_{-i} \neq (0,0)$, set
    \[
        R_i(a_{-i}) = (i - \textstyle \sum_{j \neq i} a_j) \bmod n.
    \]
    Intuitively, $i$'s reaction is to make the total of $(R_i(a_{-i}), a_{-i})$ equal to~$i$, taken modulo~$n$.
    For instance, we have $R(0,1,1) = (1,0,1)$ as $R_0(1,1) = (0 - 2) \bmod n = 1$, $R_1(0,1) = (1 - 1) \bmod n = 0$, and $R_2(0,1) = (2 - 1) \bmod n = 1$.
    In particular, at each outcome $a \neq (0,0,0)$, there is precisely one player~$i$, namely $i = (\sum_j a_j) \bmod n$, for which $R_i(a_{-i}) = a_i$. 
    Analogous to \textsc{Part~I}, there is one happy and two unhappy players;
    we again have $\mathcal{E}(R'_i,R_{-i}) \subseteq \{ (0, 0, 0) \}$ for all unilateral deviations, so $R$ supports $(0, 0, 0)$ in reaction-function equilibrium.
\end{proof}

\setcounter{section}{4}
\setcounter{theorem}{0}
\begin{theorem} 
    If a safe outcome $\ubar{a} \in A$ of the strategic game is a Nash equilibrium, %
    then there is a safe reaction-function equilibrium.
    In general, if a safe reaction-function equilibrium exists, then there exists a Pareto-efficient safe reaction-function equilibrium.
\end{theorem}

\begin{proof}
    \noindent
    \textsc{Part I:}
    This follows by Observation~\ref{OBS:constant}.
    
    \bigskip \noindent
    \textsc{Part II:}
    Let profile $R$ be the safe reaction-function equilibrium that supports outcome~$\hat{a}$.
    Suppose that $\hat{a}$ can be Pareto improved, say by outcome~$a^*$.
    Define profile $R^*$ by matching $R$ everywhere except for $R^*_i(a^*_{-i}) = a^*_i$ for each player~$i$.
    Then $a^* \in \mathcal{E}(R^*)$ and $\mathcal{E}(R^*) \subseteq \mathcal{E}(R) \cup \{ a^* \}$.
    In particular, for each $a \in \mathcal{E}(R^*)$ and player~$i$, $u_i(a^*) \geq u_i(\hat{a}) \geq u_i(a)$.
    Hence, $R^*$ is unambiguous.
    As $u_i(a^*) \geq u_i(\hat{a}) \geq \ubar{v}_i$ and $R^*_i$ otherwise matches the safe~$R_i$, also each $R^*_i$ is safe.
    If, for contradiction, player $i$ deviates at $R^*_i$ through a constant $R'_i$ with outcome $a' \in \mathcal{E}(R'_i,R^*_{-i})$, then $a' \in \mathcal{E}(R'_i,R_{-i})$ as well.
    Hence, $i$ would deviate at $R$, which is a contradiction.
    Hence, $R^*$ is a safe reaction-function equilibrium supporting the outcome $a^*$ that Pareto improves on~$a$.
    If $a^*$ is not Pareto efficient, we can repeat the argument with respect to whichever outcome Pareto dominates~$a^*$.
\end{proof}

\setcounter{section}{5}
\setcounter{theorem}{0}
\setcounter{observation}{2}

\begin{observation}
    For symmetric investment games, a monotone reaction function is norm-proof if and only if it exhibits conditional collaboration.
\end{observation} 

\begin{proof}
    \noindent
    \textsc{Part I:}
    Let $R^*_i \in \mathcal{M}_i$ be norm-proof.
    Fix an amount $\alpha \in \{ 0, \dots, H \}$ and let $R \in \mathcal{N}$ be such that, for a generic player~$i$, $R_i(a_{-i}) = \alpha$ if $\min_{j \neq i} a_j \geq \alpha$ and otherwise $R_i(a_{-i}) = 0$.
    Then $(\alpha, \dots, \alpha) \in \mathcal{E}(R)$ and $U_i(R^*_i, R_{-i}) \geq U_i(R)$ requires $R^*_i(\alpha, \dots, \alpha) = \alpha$.
    As this applies for each $\alpha \in \{ 0, \dots, H \}$, we conclude that $R^*_i$ must coordinate if all others do.
    Next, let $a_{-i} \in A_{-i}$ be such that $\alpha \equiv \min_{j \neq i} a_j < \max_{j \neq i} a_j \equiv \beta$.
    As $R^*_i$ is monotone, $\alpha = R^*_i(\alpha, \dots, \alpha) \leq R_i^*(a_{-i}) \leq R^*_i(\beta, \dots, \beta) = \beta$.
    Hence, $R^*_i$ exhibits conditional collaboration.
    
    \bigskip \noindent
    \textsc{Part II:}
    Let $R^*_i$ exhibit conditional collaboration.
    For contradiction, suppose $R^*_i$ is not norm-proof, so there is $R \in \mathcal{N}$ for which $U_i(R^*_i, R_{-i}) < U_i(R)$.
    We claim that each $a \in \mathcal{E}(R)$ is coordinated:
    Without loss, say $a \in \mathcal{E}(R)$ is ordered such that $a_1 \leq \dots \leq a_n$.
    For contradiction, suppose $a_1 < a_n$.
    Then $a_{-1} \geqq a_{-n}$.
    As $R_i$ is monotone and $a \in \mathcal{E}(R)$, $a_1 = R_1(a_{-1}) \geq R_n(a_{-n}) = a_n$, which is a contradiction.
    Hence, $U_i(R) = u_i(\alpha, \dots, \alpha)$ for some $\alpha \in \{ 0, \dots, H \}$.
    But $R^*_i(\alpha, \dots, \alpha) = \alpha$, so $U_i(R^*_i, R_{-i}) \geq u_i(\alpha, \dots, \alpha) = U_i(R)$.
    This is a contradiction.
\end{proof}

\begin{theorem} 
    For weakest-link games, reaction function $R_i \in \mathcal{R}_i$ is norm-proof and payoff consistent if and only if $R_i = \textit{BR}_i$.
\end{theorem}

\begin{proof}
    Let $R_i \in \mathcal{R}_i$ be norm proof and payoff consistent. 
    Consider an arbitrary $a \in A$ and let $\alpha = \min_{j \neq i} a_j$.
    By payoff consistency, $R_i(a_{-i}) = R_i(\alpha, \dots, \alpha)$.
    By Observation~\ref{OBS:normproof}, as $R_i$ is norm-proof, $R_i(\alpha, \dots, \alpha) = \alpha$.
    Hence, $R_i(a_{-i}) = \min_{j \neq i} a_j = \textit{BR}_i(a_{-i})$.
\end{proof}

\begin{theorem} 
    For high-risk weakest-link games, reaction function $R_i \in \mathcal{R}_i$ is weakly dominant against monotone play if and only if $R_i = \textit{BR}_i$.
\end{theorem}

\begin{proof}
    \textsc{Part I:}
    We show first that $\textit{BR}_i$ is weakly dominant against monotone play.
    For contradiction, suppose not.
    Then there is $R \in \mathcal{R}$ and $a \in \mathcal{E}(R)$ such that $u_i(a) = U_i(R) > U_i(\textit{BR}_i,R_{-i})$.
    We claim that, as $\textit{BR}_i$ and $R_{-i}$ are monotone, $\mathcal{E}(\textit{BR}_i,R_{-i}) \neq \emptyset$.

    To see this, let $R \in \mathcal{R}$ be such that each $R_i$ is monotone. 
    Let $a^0 = (H,\dots,H)$, $a^1 = R(a^0)$, $a^2 = R(a^1)$, and so on.
    This defines a monotone sequence, $a^0 \geqq a^1 \geqq \dots$, on a finite set $A$ bounded below by $(0,\dots,0)$.
    At some step~$k$, we have $a^k = a^{k+1} = R(a^k)$, so $a^k \in \mathcal{E}(R) \neq \emptyset$.
    
    Next, as reaction $\textit{BR}_i$ is safe, $U_i(\textit{BR}_i,R_{-i}) \geq \ubar{v}_i = u_i(0,\dots,0) = 0$.
    Hence, $u_i(a) > 0$.
    With high-risk payoffs, $\lambda < H / (H-1)$, this requires $a_i \leq \min_{j \neq i} a_j = \textit{BR}_i(a_{-i})$.
    To see this, if we had $a_i > \min_{j \neq i} a_j$, then
    \[
        \textstyle
        u_i(a) 
        = \lambda \cdot \min_j a_j - a_i
        \leq \lambda \cdot \min_j a_j - (\min_j a_j + 1)
        \leq \lambda \cdot (H - 1) - H
        < 0. 
    \]
    Hence, $a_i \leq \textit{BR}_i(a_{-i})$.
    If this holds with equality, then $a \in \mathcal{E}(\textit{BR}_i,R_{-i})$ as well, which would contradict $u_i(a) > U_i(\textit{BR}_i,R_{-i})$.
    Hence, $a_i < \textit{BR}_i(a_{-i})$.

    Define now $a^0 = (a^0_i,a_{-i}) \in A$ with $a^0_i = \textit{BR}_i(a_{-i})$.
    Hence, $a^0 \geqq a$ with $a^0_i > a_i$.
    For $j \neq i$, as $R_j$ is monotone and $a \in \mathcal{E}(R)$, $a^1_j \equiv R_j(a^0_{-j}) \geq R_j(a_{-j}) = a_j = a^0_j$.
    Let $a^1_i = \textit{BR}_i(a^0_{-i}) = a^0_i$.
    Hence, $a^1 \geqq a^0 \geqq a$ with $a^1_i > a_i$.
    For $k \in \mathbb{N}$, define $a^{k+1} \in A$ through $a^{k+1} = (\textit{BR}_i,R_{-i})(a^k)$.
    This defines a monotone sequence, $a^0 \leqq a^1 \leqq \dots$, that is bounded above by $(H,\dots,H)$ and thus converges at some step $k$ with $a^k = a^{k+1} = (\textit{BR}_i,R_{-i})(a^k)$.
    That is, $a^k \geqq a$ is a fixed point of $(\textit{BR}_i,R_{-i})$ with $a^k_i > a_i$.
    Hence, $a^k_i = \textit{BR}_i(a^k_{-i}) = \min_j a^k_j$.
    As noted above, $a_i = \min_j a_j$.
    But then $u_i(a) = (\lambda - 1) a_i < (\lambda - 1) a^k_i = u_i(a^k) \leq U_i(\textit{BR}_i,R_{-i})$, which is a contradiction.

    \bigskip \noindent
    \textsc{Part II:}
    We now show uniqueness.
    Let $R_i \in \mathcal{R}_i$ be weakly dominant against monotone play, so, in particular, $U_i(R) \geq U_i(\textit{BR}_i,R_{-i})$ for all monotone $R_{-i} \in \mathcal{R}_{-i}$.
    Consider an arbitrary $a_{-i} \in A_{-i}$ and let $R_{-i} \in \mathcal{R}_{-i}$ be the constant (and thus monotone) reactions in which each player $j \neq i$ always takes action~$a_j$.
    If $R_i(a_{-i}) > \textit{BR}_i(a_{-i})$, then
    \[
        \textstyle U_i(R) = \lambda \cdot \min_j a_j - R_i(a_{-i}) < \lambda \cdot \min_j a_j - \textit{BR}_i(a_{-i}) = U_i(\textit{BR}_i,R_{-i}),
    \]
    which is a contradiction.
    The same applies if $R_i(a_{-i}) < \textit{BR}_i(a_{-i})$ as then
    \[
        \textstyle U_i(R) = (\lambda - 1) \cdot R_i(a_{-i}) < (\lambda - 1) \cdot \textit{BR}_i(a_{-i}) = U_i(\textit{BR}_i,R_{-i}).
    \]
    Hence, we must have $R_i(a_{-i}) = \textit{BR}_i(a_{-i})$ for each $a_{-i} \in A_{-i}$, so $R_i = \textit{BR}_i$.
\end{proof}

\begin{theorem} \label{TH:PGchar}
    For public-good games, reaction function $R_i \in \mathcal{M}_i$ is norm-proof and payoff consistent if and only if there is a monotone function $f_i \colon \mathbb{R} \to \mathbb{Z}$ such that, %
    for each $a_{-i} \in A_{-i}$,
    \[
        R_i(a_{-i}) 
        = f_i \left (\frac{1}{n-1} \textstyle \sum_{j \neq i} a_j \right )
        \in \left \{
            {\left \lfloor \frac{1}{n-1} \textstyle \sum_{j \neq i} a_j \right \rfloor}, 
            {\left \lceil \frac{1}{n-1} \textstyle \sum_{j \neq i} a_j \right \rceil}
            \right \}.
    \]
\end{theorem}

\begin{proof}
    Let $R_i \in \mathcal{M}_i$ be norm-proof and payoff consistent.
    Consider an arbitrary $a \in A$ and let 
    \[
        m = \left \lfloor \frac{1}{n-1} \textstyle \sum_{j \neq i} a_j \right \rfloor.
    \]
    We can ``flatten'' $a$ to $a' \in A$ with $a'_{-i} = (m,\dots,m,m+1,\dots,m+1)$ and $\sum_{j \neq i} a'_j = \sum_{j \neq i} a_j$.
    By payoff consistency, $R_i(a_{-i}) = R_i(a'_{-i})$.
    If $a'_{-i} = (m, \dots, m)$, then, by Observation~\ref{OBS:normproof} as $R_i$ is norm-proof, $R_i(a'_{-i}) = m$ and thus $R_i(a_{-i}) = m$, as desired.
    If not, so $a'_{-i} \neq (m,\dots,m)$, then $a'_{-i}$ contains at least one $m$ and at least one $m+1$.
    By monotonicity, $R_i(m,\dots,m) \leq R_i(a'_{-i}) \leq R_i(m+1,\dots,m+1)$.
    By Observation~\ref{OBS:normproof}, $R_i(m,\dots,m) = m$, $R_i(m+1,\dots,m+1) = m+1$, and thus $R_i(a'_{-i}) \in \{m, m+1\}$.       Hence, $R_i(a_{-i}) \in \{m, m+1\}$, as desired.
\end{proof}

\begin{theorem} \label{TH:PG}
    For high-risk public-good games, $R_i \in {\cal M}_i$ is welfare-maximizing conditional on safe play if and only if  \[
        R_i^*(a_{-i}) = \left \lfloor \frac{1}{n-1} \textstyle \sum_{j \neq i} a_j \right \rfloor.
    \]
\end{theorem}

\begin{proof}
    As $\lambda > 1/n$, welfare is increasing in player $i$'s investment~$a_i$:
    \[
        \textstyle \sum_j u_j(a) = n \lambda \sum_j a_j - \sum_j a_j.
    \]
    Hence, we only need to identify the highest safe investment.
    We will show that this is given by $R^*_i(a_{-i})$.
    By definition, the reaction $a_i = R_i(a_{-i})$ to actions $a_{-i}$ is safe whenever $u_i(a) = \lambda \sum_j a_j - a_i \geq \ubar{v}_i = 0$.
    Equivalently,
    \[
        \textstyle
        a_i \leq \lambda \sum_j a_j 
        \iff (1 - \lambda) a_i \leq \lambda \sum_{j \neq i} a_j 
        \iff a_i \leq \displaystyle \frac{\lambda}{1 - \lambda} \textstyle \sum_{j \neq i} a_j.
    \]
    As $\lambda > 1/n$, $\lambda / (1 - \lambda) > 1 / (n-1)$.
    Hence, $R^*_i$ is safe:
    \[
        R^*_i(a_{-i}) 
        = \left \lfloor \frac{1}{n-1} \textstyle \sum_{j \neq i} a_j \right \rfloor
        \leq \displaystyle \frac{1}{n-1} \textstyle \sum_{j \neq i} a_j 
        < \displaystyle \frac{\lambda}{1 - \lambda} \textstyle \sum_{j \neq i} a_j.
    \]
    Let now $a_i > R^*_i(a_{-i})$, that is, $a_i$ exceeds the average.
    Then $\sum_{j \neq i} a_j \leq (n-1) a_i - 1$.
    For high-risk public-good games, $\lambda < H / (nH - 1)$.
    As the fraction $x / (n x - 1)$ is decreasing in~$x$, we have $H / (nH - 1) \leq a_i / (n a_i - 1)$.
    Furthermore,
    \[
        \frac{\lambda}{1 - \lambda} < \frac{a_i}{(n-1) a_i - 1}.
    \]
    Hence,
    \[
        \frac{\lambda}{1 - \lambda} \textstyle \sum_{j \neq i} a_j
        < \displaystyle \frac{a_i}{(n-1) a_i - 1} \cdot \left ( (n-1) a_i - 1 \right )
        = a_i.
    \]
    As desired, this shows that $a_i > R^*_i(a_{-i})$ is not safe.
    Hence, $R^*_i(a_{-i})$ is the highest safe investment.
\end{proof}

\section*{Appendix: Further remarks on Theorem~\ref{TH:moreplayers}} \label{APP:remarks}

Suppose there are three players with two actions each, $A_i = \{0,1\}$.
Suppose that, for each player~$i$ and outcome $a \neq (1,1,1)$, $u_i(a) > u_i(1,1,1)$.
By contradiction, suppose that we can support $(1,1,1)$ in reaction-function equilibrium.
As $(1,1,1)$ is uniquely worst, there cannot be any other fixed point.
Moreover, if there is an outcome $a$ for which $R(a)$ differs from $a$ in only one coordinate, say $R_i(a_{-i}) \neq a_i$, then player $i$ beneficially deviates through the constant reaction $R'_i$ associated with action $a_i$ as $a \in \mathcal{E}(R'_i,R_{-i})$.
Hence, for each $a \neq (1,1,1)$, the reaction $R(a)$ must differ from $a$ in at least two coordinates.
We use this observation in each step below:
\begin{enumerate}[noitemsep]
    \item For the reaction $R(0,0,0)$ to differ in at least two coordinates, there are at least two players $i$, without loss players $1$ and $2$, with $R_i(0,0) = 1$.
    \item As $R_1(0,0) = 1$, we must have $R(1,0,0) = (1,1,1)$.
    That is, $R_3(1,0) = 1$.
    \item As $R_3(1,0) = 1$, we must have $R(1,0,1) = (0,1,1)$.
    That is, $R_1(0,1) = 0$.
\end{enumerate}
In Step 1, we concluded that $R_1(0,0) = 1$, but also that $R_2(0,0) = 1$.
Hence, we can take Steps 2 and 3 from player $2$'s perspective as well.
By symmetry, going through outcomes $(0,1,0)$ and $(0,1,1)$, we get to $(0,0,1)$ and conclude that $R_2(0,1) = 0$.
But then $R(0,0,1) = (0,0,\cdot)$, which differs in at most one coordinate.
This is a contradiction.

\section*{Appendix: Simulation study}
    We complement the theoretical findings with simulations.
    The setup matches \citet{Fischbacheretal2001}, namely a four-player public-goods game with actions $A_i = \{0, \dots, 20\}$ and payoffs
    \[
        \textstyle u_i(a) = 20 - a_i + 0.4 \sum_j a_j.
    \]
    We restrict to monotone, payoff-consistent reaction functions.
    These can conveniently be represented by non-decreasing functions $R_i \colon \{ 0, \dots, 60 \} \to \{ 0, \dots, 20 \}$, where $i$ commits to contribute $R_i(\alpha)$ conditional on the total contribution of the three other players being~$\alpha$.
    We maintain a list of $500$ such functions, which evolves as worst-performing reactions get replaced by ``mutated'' variants of the best-performing ones.
    Initially, the list consists solely of pure free-riders, but through experimentation, new and more effective functions emerge and gradually displace the original ones.
    Once this process converges, the theoretical prediction is that most reaction functions will be alike the norm-proof $R^*_i$ (Theorem~\ref{TH:PGchar}).
    \emph{The main take-away from the simulation study is that it gives further support to this prediction.}
    
    The flow of the simulation is as follows:
    \begin{enumerate}[noitemsep]
        \item There are $100$ independent runs, each starting with a list of pure free-riding.
        \item A run is composed of $15\ 000$ batches of games.
        Mutation occurs after each batch based on average payoffs within the batch.
        \item A batch is composed of $10\ 000$ games.
        In each game, four reaction functions are drawn at random from the list (with replacement) and payoffs of the highest fixed point are realized.
        As reactions are monotone, this is easy to compute.
    \end{enumerate}
    In total, the simulations cover $100 \times 15\ 000 \times 10\ 000 = 15$ billion games.
    
    To capture convergence, mutation frequency and impact decreases exponentially.
    That is, the number of replaced worst-performing reaction functions is smaller in later batches.
    Moreover, the replacements become more similar to the top-performing reaction functions.
    (If $R_i$ is the top-performing reaction to copy, then the mutated $R'_i$ is such that, for each amount $\alpha \in \{0, \dots, 60\}$, we draw $\varepsilon$ from $\{ -1, 0, 1 \}$ and set $R'_i(\alpha) = R_i(\alpha) + \varepsilon$.
    By making $\varepsilon = 0$ more likely, $R'_i$ becomes more similar to $R_i$.
    If needed, we modify $R'_i$ to make it monotone and within bounds, $0 \leq R'_i(\alpha) \leq 20$.)
    
    Although the functions are defined for all $\alpha \in \{0, \dots, 60\}$, mainly higher values of $\alpha$ tend to be reached during play in later batches (i.e., at the highest fixed point).
    As a result, the replacement step is disproportionately focused on these higher values and the function's behavior for smaller values, say $\alpha < 45$, is likely under-optimized and a bit more ``noisy''.
    
    \begin{figure}[!htb]
        \centering
        \begin{tikzpicture}
            \begin{axis}[
                width=14cm,
                height=9cm,
                grid=major,
                xlabel={Total contribution $\alpha \in \{ 0, \dots, 60 \}$ by others},
                ylabel={Contribution $R_i(\alpha) \in \{ 0, \dots, 20 \}$},
                xmin=0, xmax=60,
                ymin=0, ymax=20,
            ]
            
            \addplot+[no markers, color=gray, densely dashed] table [col sep=comma, x index = 0, y index = 1] {RFE-sims.csv};
            \addplot+[only marks, mark=o, mark size=2pt, color=black] table [col sep=comma, x index = 0, y index = 2] {RFE-sims.csv};
            \addplot+[no markers, color=black, line width=1pt] table [col sep=comma, x index = 0, y index = 3] {RFE-sims.csv};
            \addplot+[only marks, mark=*, mark size=1pt, color=black] table [col sep=comma, x index = 0, y index = 4] {RFE-sims.csv};
            \end{axis}
        \end{tikzpicture}
        \caption{%
            The dashed line is to match the average contribution of the others.
            Open circles is $R^*_i$, to round down the average.
            The solid curve shows simulation averages and filled dots are simulation median values.
            Simulation medians match $R^*_i$ whenever the dot is encircled.}
        \label{FIG:sims}
    \end{figure}
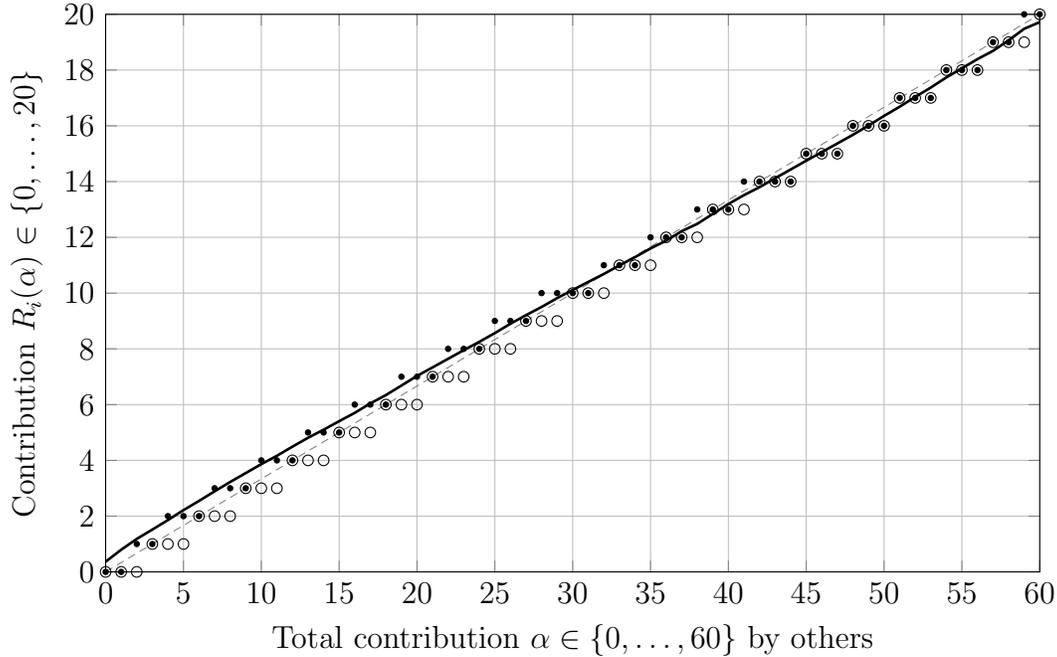
    
    Figure~\ref{FIG:sims} summarizes the $100 \times 500 = 50\ 000$ reaction functions that remain at the end of the runs.
    The solid curve shows the average across these functions.
    It tracks the average contribution of the others (dashed line) and, importantly, is consistently below it for high~$\alpha$. 
    Figure~\ref{FIG:sims} also includes ``round down the average'' $R^*_i$ (open circles) and simulation medians (filled dots).
    These match up remarkably well in the higher end (except at $\alpha = 59$).
\end{document}